\documentclass[technicalreport]{iwsda15}
\usepackage[dvips]{graphicx}
\usepackage{times}
\usepackage{mathrsfs}
\usepackage{amsthm}
\usepackage{makeidx}
\usepackage{examples}
\usepackage{amssymb}
\usepackage{graphicx}
\usepackage{epstopdf}
\usepackage{color}
\usepackage{url}
\graphicspath{{../eps/}{../ps/}}
\usepackage{psfrag}    
\usepackage{amsmath}
\font\msbm=msbm10 at 10pt
\newcommand{\ZZ}{\mbox{\msbm Z}}

\newcommand{\NN}{\mbox{\msbm N}}

\newcommand{\FF}{\mbox{\msbm F}}
\def \Z {{\ZZ}}

\def \N {{\NN}}
\def \F {{\FF}}

\newtheorem{theorem}{Theorem}
\newtheorem{lemma}[theorem]{Lemma}
\newtheorem{remark}[theorem]{Remark}
\newtheorem{example}[theorem]{Example}
\newtheorem{definition}[theorem]{Definition}
\newtheorem{algorithm}{Algorithm}

\title{On Dress Codes with Flowers}%

\authorlist{
 \authorentry[201221007@daiict.ac.in]{Krishna Gopal Benerjee}{labelA}  
 \authorentry[mankg@computer.org]{Manish Kumar Gupta}{labelB}
}
\affiliate[labelA]{Dhirubhai Ambani Institute of Information and Communication Technology, Gandhinagar, Gujarat, 382007 India}

\begin{document}

\maketitle

\begin{abstract}\textit{Fractional Repetition} (FR) codes are well known class of \textit{Distributed Replication-based Simple Storage (Dress)} codes for the Distributed Storage Systems (DSSs). In such systems, the replicas of data packets encoded by Maximum Distance Separable (MDS) code, are stored on distributed nodes. Most of the available constructions for the FR codes are based on combinatorial designs and Graph theory. In this work, we propose an elegant sequence based approach for the construction of the FR code. In particular, we propose a beautiful class of codes known as Flower codes and study its basic properties. 
\end{abstract}

\section{Introduction}
In Distributed Storage Systems (DSSs), data file is encoded into certain packets and those packets are distributed among $n$ nodes. A data collector has to collect packets from any $k$ (called reconstruction degree) nodes among the $n$ nodes to reconstruct the whole file. In the case of node failure, system is allowed to reconstruct a new node to replace the failed node. The repair is of two types viz. exact and functional. In the repair process, the new node is constructed by downloading $\beta$ packets from each node (helper node) of a set of $d$ (repair degree) nodes. Thus total bandwidth for a repairing a node is $d\beta$. For an $(n, k, d)$DSS, such regenerating codes are specified by the parameters  $\{[n, k, d], [\alpha,\beta, B] \}$, where B is the size of the file and $\alpha$ is the number of packets on each node. 
One has to optimize both $\alpha$ and $\beta$, hence we get two kind of regenerating codes viz. Minimum Storage Regenerating (MSR) codes useful for archival purpose and Minimum Bandwidth Regenerating (MBR) codes useful for Internet applications \cite{rr10,survey}. Some of the MBR codes studied by the researchers fails to optimize other parameters of the system such as disk I/O, computation and scalability etc. Towards this goal, a class of MBR codes called Dress codes were introduced and studied by researchers \cite{rr10,6062413,RSKR10,7066224} to optimize disk I/O. These codes have a repair  mechanism known as encoded repair or table based repair. Dress codes consisting of an inner code called fractional repetition (FR) code  and outer MDS code. Construction of FR codes has been an important research problem and many constructions of FR codes are known based on graphs \cite{DBLP:journals/corr/abs-1102-3493,rr10,Wangwang12,DBLP:journals/corr/SilbersteinE14,DBLP:journals/corr/SilbersteinE15,7118709}, combinatorial designs \cite{6912604,DBLP:journals/corr/OlmezR14,DBLP:journals/corr/abs-1210-2110,DBLP:journals/corr/abs-1210-2110,6763122,7004501,DBLP:journals/corr/SilbersteinE14,DBLP:journals/corr/abs-1210-2110} and other combinatorial configuration \cite{6810361,6033980,6570830,DBLP:journals/corr/abs-1208-2787}. Existence of FR code is discussed in \cite{DBLP:journals/corr/abs-1201-3547}. FR codes have been studied in different directions such as Weak FR codes \cite{DBLP:journals/corr/abs-1302-3681,DBLP:journals/corr/abs-1303-6801}, Irregular FR codes \cite{6804948}, Variable FR codes \cite{6811237}. A new family of Fractional Repetition Batch Code is studied in \cite{DBLP:journals/corr/Silberstein14,DBLP:journals/corr/SilbersteinE15}. For a given FR code, algorithms to calculate repair degree $d$ and upper-bound of reconstruction degree $k$, is given in \cite{DBLP:journals/corr/abs-1305-4580}. 
In this paper, a more general definition of FR code (Definition \ref{defWFR}) is considered for a realistic practical scenario. Further, a sequence based construction for FR code, is given. In particular,  construction of Flower code is studied in detail. 

The structure of the paper is as follows. In Section $2$, we define a general FR codes and collects relevant background material. In Section $3$, construction of flower code with single and multiple ring 
is given. We also study its relations with sequences. In particular, we show how to construct them using arbitrary binary sequences. Final section concludes the paper with general remarks.

\begin{figure}
\centering
\includegraphics[scale=0.3]{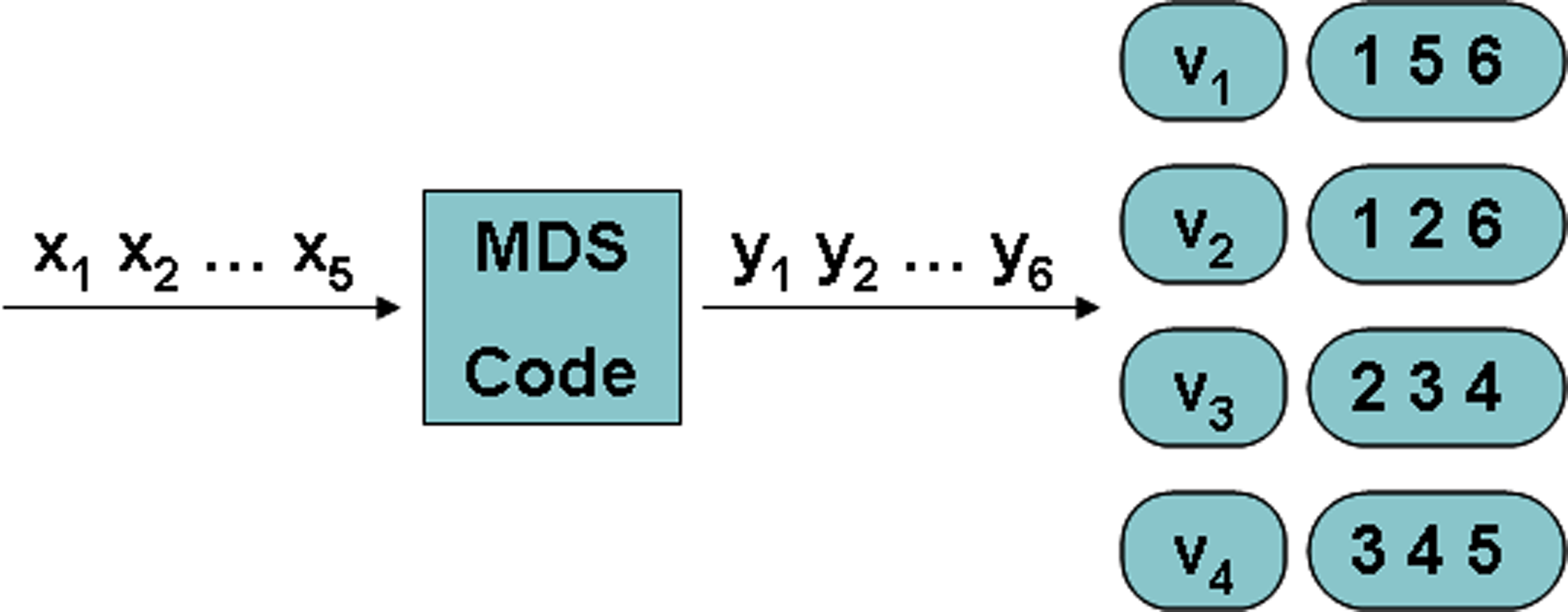}
\caption{DRESS Code consisting of an inner fractional repetition code $\mathscr{C}$ having $n=4$ nodes, number of packets $\theta=6$, replication factor $\rho=2$ and repair degree $d=3$ and an outer MDS code.}
\label{dressstrong}
\end{figure}

\section{Background}
Distributed Replication-based Simple Storage (Dress) Codes consists of an inner Fractional Repetition (FR) code and an outer MDS code (see Figure \ref{dressstrong}). In this, $5$ packets are encoded into $6$ packets using a MDS code and then each packet is replicated twice and distributed among $4$ nodes. In case of one node failure data can be recovered easily in such a system. 
We formally define FR code as follows.



\begin{definition} (Fractional Repetition Code): 
 Given a DSS with $n$ nodes $U_i  (i \in\Omega_n = \{1,2,\ldots,n\})$ and $\theta$ packets $P_j ( j \in \Omega_{\theta} = \{1,\ 2,\ldots,\theta \})$, one can define FR code  $\mathscr{C}(n, \theta, \alpha, \rho)$ as a collection $\mathscr{C}$ of $n$ subsets $U_i$ ($i\in\Omega_n$) of a set $\{P_j:j\in\Omega_{\theta}\}$, which satisfies the following conditions.
\begin{itemize}
	\item For each $j\in\Omega_{\theta}$, packet $P_j$ appears exactly $\rho_j$ ($\rho_j\in\N$) times in the collection $\mathscr{C}$.
	\item For every $i = 1, 2, \ldots, n$, $|U_i|=\alpha_i$ ($\alpha_i \in\N$),
\end{itemize}
where $\alpha_i$ denotes the number of packets on the node $U_i$,   $\rho = \max\{\rho_j\}_{j=1}^{\theta}$ is the maximum replication among all packets. The maximum number of packets on any node is given by $\alpha=\max\{\alpha_i\}_{i=1}^n$.  Clearly $\sum_{j=1}^{\theta}\rho_j = \sum_{i=1}^n\alpha_i$. 
\begin{remark}
When a node fails it can be repaired by a set of different nodes. Number of nodes contacted for repairing the fail node is known as repair degree. 
Hence each node $U_i$ has a set of different repair degrees. Let $(d_i)^j$ denotes the typical repair degree of the node $U_i$ then the set of repair degrees is $\{(d_i)^j : j \in \N \}$. If  
$d_i=\max_{j} \{(d_i)^j\}$ is the maximum repair degree of the node $U_i$ then $d=\max\{d_i\}_{i=1}^n$ denotes the maximum repair degree of any node.
\end{remark}

\label{defWFR}
\end{definition}

An example of FR code $\mathscr{C}: (7, 5, 3, 4)$ is shown in Table \ref{nodep4}. In this example, packet $1$ is replicated $4$ times and all other packets are replicated $3$ times. The set of 
repair degrees for each node are shown in the last column. Note that node $U_7$ has two different repair degrees.

\begin{table}[ht]
\caption{Node-Packet Distribution for FR code $\mathscr{C}:\ (7, 5, 3, 4)$.}
\centering 
\begin{tabular}{|c||c|c|}
\hline
\textbf{Nodes} &\textbf{Packets distribution} & \textbf{Set of repair degrees} \\[0.5ex] 
\textbf{$U_i$} &     &\textbf{ for each node}\\ 
\hline\hline
$U_1$& $P_1,\ P_5$       &\{2\}\\ 
\hline
$U_2$& $P_1,\ P_2$       &\{2\}\\ 
\hline 
$U_3$& $P_2,\ P_3$       &\{2\}\\ 
\hline 
$U_4$& $P_1,\ P_3,\ P_4$ &\{3\}\\ 
\hline 
$U_5$& $P_2,\ P_4,\ P_5$ &\{3\}\\ 
\hline
$U_6$& $P_3,\ P_5$      &\{2\}\\ 
\hline
$U_7$& $P_1,\ P_4$           &\{1,2\}\\ 
\hline
\end{tabular}
\label{nodep4}
\end{table}

\begin{remark}
Each FR code $\mathscr{C}: (n, \theta, \alpha, \rho)$ considered in this paper from now on has same replication factor $\rho$ for each packet $i.e.$ $\rho_j=\rho$ for every $j=1,2,\ldots, \theta$.
\end{remark}

For every FR code one can generate a node-packet distribution incidence matrix denoted by $M_{n\times\theta}$. The distribution matrix is unique representation of FR code. The formal definition of node-packet distribution incidence matrix is as follows.

\begin{definition} (Node-Packet Distribution Incidence Matrix): For an FR code $ \mathscr{C}: (n, \theta, \alpha , \rho)$, a node-packet distribution incidence matrix is a  matrix $M_{n \times \theta}$ = $[a_{ij}]_{n \times \theta}\ s.t.$ 
\[
a_{ij} = \left\{ 
\begin{array}{ll}
a_{ij} & :\mbox{ if packet $P_j$ appears $a_{ij}$ times on node $U_i$;} \\
0 & :\mbox{ if packet $P_j$ is not on node $U_i$.}\\
\end{array}
\right.
\]
\label{incidence matrix}
\end{definition}
Clearly, if the packet $P_j$ appears exactly ones on node $U_i$ then the matrix is binary. 
For example, the node packet distribution incidence matrix $M_{7\times 5}$ for the FR code $\mathscr{C}:\ (7, 5, 3, 4)$ as given in Table \ref{nodep4}, will be
\[
M_{7 \times 5} =
\begin{bmatrix}
1 & 0 & 0 & 0 & 1  \\
1 & 1 & 0 & 0 & 0  \\
0 & 1 & 0 & 1 & 0  \\
1 & 0 & 1 & 1 & 0  \\
0 & 1 & 0 & 1 & 1  \\
0 & 0 & 1 & 0 & 1  \\
1 & 0 & 0 & 1 & 0  \\
\end{bmatrix}. 
\]
In order to construct an FR code $\mathscr{C}: (n, \theta, \alpha , \rho)$, one has to drop $\theta$ packets on $n$ nodes such that an arbitrary packet $P_j$ is replicated $\rho$ times in the code. The size of node $U_i$ is $\alpha_i$. 
 This motivates us to define $cycle\ m$ and $jump$.

\begin{definition} (Cycle): For every $1 \leq m \leq \rho$, a cycle is defined as one complete dropping of packets on $n$ nodes such that all $\theta$ packets are exhausted from $\Omega_{\theta} = \{1,\ldots,\theta \}$ without any replication. Note that $\rho$ is the replication factor of each packets in FR code $\mathscr{C}: (n, \theta, \alpha , \rho)$.
\label{cycle}
\end{definition}

An example of such cycle is shown in Table \ref{jump and cycle}.

\begin{definition} (Jump): A jump is defined as the number of null packets between two consecutive packets from $\Omega_{\theta} = \{1,\ldots,\theta \}$ while dropping them on $n$ nodes $U_1,U_2, \ldots ,U_n$.
\end{definition}

An example of such jump is shown in Table \ref{jump and cycle}. In the table, dash between two packets (such as packets indexed by $3$ and $4$) represents jump for the certain cycle.
\begin{table}[ht]
\caption{Jump and cycles for the $4$ packets on $3$ nodes.} 
\begin{center}
 {\renewcommand\arraystretch{1.5}
\begin{tabular}{|c||c|c|c|c|c|c|c|c|c|c|c|c|}\hline
   \textbf & \multicolumn{5}{c|}{\textbf{Cycle 1}} & \multicolumn{7}{c|}{\textbf{Cycle 2}} \\ \cline{2-13}
\hline\hline
Packet Index & 1 & 2 & 3 & - & 4 & - & - & 1 & 2 & - & 3 & 4 \\
\hline
Node Index   & 1 & 2 & 3 & 1 & 2 & 3 & 1 & 2 & 3 & 1 & 2 & 3 \\
\hline
\end{tabular}}
\end{center}
\label{jump and cycle}
\end{table}

One can associate  a binary characteristic sequence (dropping sequence), node sequence (ordering the node sequence where the packet is dropped) and node-packet incidence matrix with any   fractional repetition code. We now formally collect these definitions. 


\begin{definition} (Dropping Sequence): An FR code $\mathscr{C}: (n, \theta, \alpha , \rho)$ can be characterized by a binary characteristic sequence (weight of the sequence is $\rho\theta$) which is one whenever a packet is dropped on a node and zero whenever no packet is dropped. 
\label{dropping sequence}
\end{definition}

\begin{example}
For a given $n=4$ and $\theta=6$ a possible dropping sequence for the FR code $\mathscr{C}: (4, 6, 3, 2)$ as shown in Figure \ref{dressstrong}, is $\left\langle 1,1,1,1,1,1,0,0,0,1,1,1,0,0,1,1,1 \right\rangle$. Using the sequence one can generate the FR code by dropping packet $P_1$ on node $U_1$, since $d(1)=1$ and so on.
\end{example}

Dropping sequence $\left\langle d(m)\right\rangle_{m=1}^l$ of an FR code $\mathscr{C}: (n, \theta, \alpha , \rho)$ has the following properties.
\begin{enumerate}
	\item 
	\begin{equation*}
 d(m) = \left\{ 
\begin{array}{ll}
1  & :\mbox{ if packet is dropped on node $U_t$;} \\
0  & :\mbox{ if packet is not dropped on node $U_t$,}
\end{array}
\right. \\
\end{equation*}
where
\[
 t = \left\{ 
\begin{array}{ll}
m \ (mod \ n)  & :\mbox{ if } n\nmid m; \\
n & :\mbox{ if } n\mid m.
\end{array}
\right.
\]
	\item WLOG one can set $d(l)=1$ for dropping sequence $\left\langle d(m)\right\rangle_{m=1}^l$, where $l\in \N\ s.t.\ weight$ of $\left\langle d(m)\right\rangle_{m=1}^l$ = $w(\left\langle d(m)\right\rangle_{m=1}^l)=\rho\theta \mbox{ and } d(m)\in \{0, 1\}.$
	\item It is clear that $l\in\{\rho\theta, \rho\theta +1, ..., \rho\theta + (n-1)^{\rho\theta -1},...\}$ but one can reduce the value of $l$ such that $\rho\theta\leq l\leq\rho\theta + (n-1)^{\rho\theta -1}$ by puncturing $n$ consecutive 0's in $\left\langle d(m)\right\rangle_{m=1}^l$. It can be observed easily that if $l>\rho\theta + (n-1)^{\rho\theta -1}$ then $\exists\ q(\geq n)$ consecutive 0's in $\left\langle d(m)\right\rangle_{m=1}^l$.
	\item If $d(r)=1(1\leq r\leq l)$ then packet P$_\lambda$ is distributed on node $U_t$, where
	\begin{equation*}
 \lambda = \left\{ \,
\begin{array}{ll}
 \eta (r) & :\mbox{ if $\eta (r)\neq 0$;} \\
\theta & :\mbox{ if $\eta (r) =0$,}
\end{array}
\right.
\end{equation*}
where $\eta (r)=w\left(\left\langle d(m)\right\rangle_{m=1}^r\right) \ (mod \ \theta)$.
\end{enumerate}

\begin{definition} (Node Sequence): An FR code $\mathscr{C}: (n, \theta, \alpha , \rho)$ can be characterized by a finite sequence $\left\langle s_i \right\rangle _{i=1}^{\rho\theta}$ defined from $\left\{1, 2,..., \rho\theta \right\}\left(\subset\N\right)$ to $\Omega_n = \left\{1, 2,..., n \right\}\left(\subset\N\right)$ such as packet $P_t\left(t \in \Omega_{\theta} \right)$ is dropped on node $s_i \left( \in \Omega_n \right)$, where 
\begin{equation}
t = \left\{ \,
\begin{array}{ll}
i \ (mod \ \theta)  & :\mbox{ if } \left(i\ (mod \ \theta) \right) \neq 0; \\
\theta & :\mbox{ if } \left(i\ (mod \ \theta) \right) = 0.
\end{array}
\right.
\label{eq:example_left_right1}
\end{equation}
\label{node sequence}
\end{definition}


\begin{example}
For given $n$ nodes and $\theta$ packets, a possible node sequence is $\left\langle 1,2,3,4,1,2,2,3,4,3,4,1\right\rangle$ for FR code $\mathscr{C}: (4, 6, 3, 2)$ as shown in Figure \ref{dressstrong}.
\end{example}

\begin{remark}
For an FR code $\mathscr{C}:(n, \theta, \alpha, \rho)$, there may exist more then one dropping sequences $\left\langle d(m)\right\rangle_{m=1}^l$ as well as node sequences $\left\langle s_i \right\rangle _{i=1}^{\rho\theta}$.
\end{remark}


Next section describes the construction of the flower codes 
by arranging $n$ nodes in a single ring or multiple ring configuration and then dropping $\theta$ packets following a configuration. 
\section{Constructions of Flower Codes}
A generalized ring construction of FR codes was described in \cite{DBLP:journals/corr/abs-1302-3681} which give rise to Weak FR codes with $\rho=2$ depending upon weather $\theta$ is a multiple of $n$ or not. In this section, we generalize $\rho=2$ construction to $\rho > 2$. In order to construct the FR code of replication factor $\rho$, we first place $n$ nodes on a circle. Now we can place $\theta$ packets on each of them one by one till a cycle is complete. One has to place the packets till we complete all $\rho$ cycles. This will give rise to an FR code since all $\theta$ packets are replicated $\rho$ times in the system. Now we can vary the packet dropping mechanism by introducing jumps within a cycle or after every cycle. To do so first we define different kind of jumps. This process yields several classes of interesting FR codes.

\begin{definition} (Internal  $\&$ External Jumps): An internal jump is a jump applied within a cycle $m, 1 \leq m \leq \rho$. Similarly a jump is called external jump if it is applied between two consecutive cycles $m\; \& \;m+1, 1 \leq m \leq \rho$. 
In particular, internal (external) jump function is denoted by $f_{in} (f_{ex})$ with domain $\left\{1,2,...,\rho\theta\right\}\backslash\{\theta, 2\theta,\ldots,\rho\theta\}\left(\left\{1, 2,...,\rho\right\}\right)$. Both jump functions have common co-domain $\N\cup \{0\}$.
\label{Jumps} 
\end{definition}
\begin{example} For given nodes $n=3$, packets $\theta=4$, internal jump function $f_{in}=1$ and external jump function $f_{ex}$, one can construct node sequence $\left\langle 1,3,2,1,2,1,3,2\right\rangle$ for some FR code $\mathscr{F}_{\mathscr{C}}: (3, 4, 3, 2)$.
\end{example}

\begin{remark}
A special kind of jump, (for example, see Definition \ref{subsetjump}), can be described by a characteristic function $\xi(i)=1 \;\mbox{(drop)}\;\mbox{or} \;0\;\mbox{(do not drop)}$ which tells us when to drop a packet at a position $i, 1 \leq i \leq n$.
\end{remark}

\begin{definition} (Subset Type Jumps): \label{subsetjump}
Let $\Omega_n =\{1, 2, \ldots, n \}$ be an index set of $n$ nodes and let $A \subseteq \Omega_n$. A jump is called subset type jump if it's characteristic function $\xi(i)$ is given by 
\[
\xi(i) =\left\{
\begin{array}{ll}
1 & :\mbox{if}\; i \in A; \\
0 & :\mbox{if}\; i \notin A.
\end{array}\right.
\]
\end{definition}
Now we are ready to define a Flower code with single ring having a subset type jump within it's $\rho$ cycles. 
\begin{definition}(Flower code with single ring): A Flower code $\mathscr{F}_{\mathscr{C}}: (n, \theta, \alpha , \rho)$ with single ring  and having a subset type jumps can be defined by first placing $n$ nodes along a single ring and then dropping packets as per a subset jump $A_m (1 \leq m \leq \rho)$ (see Definition \ref{subsetjump}) within every cycle $m (1 \leq m \leq \rho)$  till we drop all $\rho \theta$  packets on the ring. An example for such jump, is illustrate in Table \ref{flower jump table}.
\label{Flower with 1 ring} 
\end{definition}

\begin{table}[ht]
\caption{A possible distribution for a Flower code $\mathscr{F}_{\mathscr{C}}: (8, 7, 4, 3)$ with subset type jumps $A_1=\{ 1, 2, 4\}$, $A_2= \{5, 6, 7, 8\}$ and $A_3= \{2, 3, 5, 6, 7\}.$} 
\begin{center}
 {\renewcommand\arraystretch{1.5}
\begin{tabular}{|c||c|c|c|c|c|c|c|c|}\hline
   \textbf{Node} & \multicolumn{7}{c|}{\textbf{Packet distribution}} & \textbf{Node Capacity} \\ \cline{2-8}
          $U_i$       & \multicolumn{3}{c|}{\textbf{For $A_1$}}&\multicolumn{2}{c|}{\textbf{For $A_2$}}&\multicolumn{2}{c|}{\textbf{For $A_3$}}&      $\alpha_i$               \\ 
\hline\hline
$U_1$& $P_1$ & $P_4$ & $P_7$ & - & - & - & - &3\\
\hline
$U_2$& $P_2$ & $P_5$ & - & - & - & $P_1$ & $P_6$& 3\\
\hline
$U_3$& - & - & - & - & - & $P_2$ & $P_7$ &2\\
\hline
$U_4$& $P_3$ & $P_6$ & - & - & - & - & - &2\\
\hline
$U_5$& - & - & - & $P_1$ & $P_5$ & $P_3$ & - &3\\
\hline
$U_6$& - & - & - & $P_2$ & $P_6$ & $P_4$ & - &3\\
\hline
$U_7$& - & - & - & $P_3$ & $P_7$ & $P_5$ & - &3\\
\hline
$U_8$& - & - & - & $P_4$ & - & - & - &1\\
\hline
\end{tabular}}
\end{center}
\label{flower jump table}
\end{table}

\begin{lemma}
Consider a Flower code $\mathscr{F}_{\mathscr{C}}: (n, \theta, \alpha , \rho)$ with single ring and having subset type jumps on node subsets $A_1, A_2,\ldots A_{\rho}$. If total number of packets distributed on node $U_{i_t} (i_t\in A_m; 1\leq m\leq\rho\mbox{ and }t=1,2,\ldots,|A_m|)$ for subset type jump with single ring on node subset $A_m$ is $P(U_{i_t},A_m)$ then 
\begin{equation*}
P(U_{i_t},A_m)=\left\{ 
\begin{array}{ll}
\left\lceil\frac{\theta}{|A_m|}\right\rceil & :\mbox{ if } t\leq\theta\ (mod\ |A_m|); \\
&\\
\left\lfloor \frac{\theta}{|A_m|}\right\rfloor & :\mbox{ otherwise.}
\end{array}
\right. 
\end{equation*}
\label{11}
\end{lemma}

\begin{proof} Suppose $A_1,A_2,\ldots,A_{\rho}$ are the subsets of $\Omega_n$. For the given $n$ nodes, $\theta$ packets and $A_m$ $(m=1,2,\ldots \rho)$ one can construct flower code $\mathscr{F}_{\mathscr{C}}: (n, \theta, \alpha , \rho)$ with single ring and subset type jump on $A_m$, $\forall$ $m$. Since subset type jump on particular subset $A_m$ is assigning $\theta$ distinct packets on nodes $U_{i_t}\in A_m$ ($1\leq t\leq |A_m|$) in some specific order. If $|A_m|$ divides $\theta$ then number of packets dropped on a particular node $U_{i_t}$ is $\frac{\theta}{|A_m|}$. If $|A_m|$ does not divide $\theta$ then after dropping $\left\lfloor\frac{\theta}{|A_m|}\right\rfloor$ packets on each node of $A_m$, there will remain $\theta\ (mod\ |A_m|)$ packets to assign nodes. Hence, there are $\theta\ (mod\ |A_m|)$ number of nodes $U_{i_t}\in A_m$ with $\left\lfloor\frac{\theta}{|A_m|}\right\rfloor+1$ packets each. Remaining nodes in the set $A_m$ have $\left\lfloor\frac{\theta}{|A_m|}\right\rfloor$ packets each. Hence, the lemma is proved.
\end{proof}

\begin{lemma}
For a Flower code $\mathscr{F}_{\mathscr{C}}: (n, \theta, \alpha , \rho)$ with single ring and having subset type jumps on node subsets $A_1, A_2,\ldots A_{\rho}$, the total number of packets stored on a particular node $U_i$ is 
\begin{equation*}
\alpha_i=\sum_{m=1}^{\rho}P(U_i,A_m),
\end{equation*}
where, $P(U_i,A_m)=0$ for $U_i\notin A_m$.
\end{lemma}
\begin{proof} For a Flower code $\mathscr{F}_{\mathscr{C}}: (n, \theta, \alpha , \rho)$ with single ring and having subset type jumps, all packets stored on a particular node are equal to the sum of total number of packets dropped on the node for each subset type jump of $A_m$ ($1\leq m\leq\rho$). Hence proved.
\end{proof}

\begin{remark}
For a Flower code $\mathscr{F}_{\mathscr{C}}: (n, \theta, \alpha , \rho)$ with a subset type jump $A_m (1 \leq m \leq \rho)$ has the following properties
\begin{itemize}
 	\item If $\left|A_{\max}\right|=\max \left\{\left|A_1\right|,\left|A_2\right|,...,\left|A_{\rho}\right|\right\}$ then
\begin{equation*}
\left\lfloor\frac{\theta}{\left|A_{\max}\right|}\right\rfloor\leq\alpha_i\leq\alpha\leq\sum\limits_{m=1}^{\rho}\left\lceil \frac{\theta}{|A_m|}\right\rceil.
\end{equation*}
	\item If $\exists\ U_i\ s.t.\ U_i\in\bigcap_{m=1}^{\rho}A_m$ then the node $U_i$ has maximum number of distributed distinct packets $i.e.\ \left|U_i\right|=\alpha$ but converse is not true.
  \item If $\exists\ U_p\left(p\in A_m\right)\ s.t.$
		\begin{equation*}
	  U_p\notin\bigcup_{\stackrel{i=1}{i\neq m}}^{\rho}A_i,\ then \  \alpha _{\max}\geq\left\lceil \frac{\theta}{\left|A_m\right|}\right\rceil
	\end{equation*}
\end{itemize}
and viz., where $1\leq i, j, p\leq n$.
\end{remark}


To construct FR code $\mathscr{C}: (n, \theta, \alpha, \rho)$, one can concatenate $\rho$ distinct cycles with some internal and external jumps, where each cycle defined on $n$ nodes. 

\begin{definition}(Flower code with multiple rings): A system with $\rho$ cycles of $\theta$ packets in which packets are distributed among $n$ nodes arranged on a circle with internal jump function $f_{in}:\left\{1,2,...,\rho\theta\right\}\backslash\{\theta, 2\theta,\ldots,\rho\theta\}\rightarrow\N\cup \{0\}$ and external jump function $f_{ex}: \left\{1, 2,...,\rho\right\}\rightarrow\N\cup\{0\}$ is called Flower code $\mathscr{F}_{\mathscr{C}}$ with parameters $n, \theta, \alpha$ and $\rho$, where $\alpha$ is the maximum collective frequency of appearance of a node in all cycles.
\label{Flower with multi ring}
\end{definition}

\begin{example} A Flower code $\mathscr{F}_{\mathscr{C}}: (n, \theta, \alpha , \rho)$ with internal jump function $f_{in}(x)=1$ and external jump function $f_{ex}(x)=0$, is illustrated in Table \ref{Flower table 3}.
\begin{table}[ht]
\caption{A Flower code $\mathscr{F}_{\mathscr{C}}: (5, 6, 3, 2)$.} 
\begin{center}
 {\renewcommand\arraystretch{1.5}
\begin{tabular}{|c||c|c|c|c|c|c|}\hline
   \textbf{Node} & \multicolumn{5}{c|}{\textbf{Packet distribution}} & \textbf{Node Capacity} \\ 
\hline\hline
$U_1$& $P_1$ & -     & $P_6$ & $P_3$ & -     &3\\
\hline
$U_2$& -     & $P_4$ & $P_1$ & -     & $P_6$ &3\\
\hline
$U_3$& $P_2$ & -     & -     & $P_4$ & -     &2\\
\hline
$U_4$& -     & $P_5$ & $P_2$ & -     & -     &2\\
\hline
$U_5$& $P_3$ & -     & -     & $P_5$ & -     &2\\
\hline
\end{tabular}}
\end{center}
\label{Flower table 3}
\end{table}
\end{example}


Note that the Flower code $\mathscr{F}_{\mathscr{C}}: (n, \theta, \alpha, \rho)$ has internal and external jump functions ($f_{in}$ and $f_{ex}$ respectively) and each Flower code $\mathscr{F}_{\mathscr{C}}: (n, \theta, \alpha, \rho)$ is an FR code $\mathscr{C}: (n, \theta, \alpha, \rho)$ so terms of node sequence $s_m$ in $\left\langle s_m\right\rangle_{m=1}^{\rho\theta}$ can be represented in terms of $f_{in}$ and $f_{ex}$ as described in Theorem \ref{theorem} .

\begin{theorem} Consider a Flower code $\mathscr{F}_{\mathscr{C}}: (n, \theta, \alpha, \rho)$ having an internal jump function $f_{in}:\left\{1,2,...,\rho\theta\right\}\backslash$ $\{\theta, 2\theta,\ldots,\rho\theta\}$ $\rightarrow\N\cup \{0\}$ and an external jump function $f_{ex}:\left\{1, 2,...,\rho\right\}\rightarrow\N\cup \{0\}$. 
If node sequence of the Flower code $\mathscr{F}_{\mathscr{C}}: (n, \theta, \alpha, \rho)$ is $\left\langle s_m\right\rangle _{m=1}^{\rho\theta}$ then
\begin{equation}
s_m = \left\{ 
\begin{array}{ll}
1 & :\mbox{ if }m=1; \\
\vartheta (m)   & :\mbox{ if } \vartheta (m) \neq 0, 1<m\leq \rho\theta; \\
n & :\mbox{ if } \vartheta (m)= 0, 1<m\leq \rho\theta;\\
0 & :\mbox{ if } m>\rho\theta,
\end{array}
\right. 
\label{Node sequence and jump}
\end{equation}
where 
\begin{equation*}
\vartheta (m)=\left[m+\sum\limits_{
     \begin{array}[b]{c}
        i=0 \\ 
        \theta\nmid i
     \end{array}
}^{m-1}f_{in}(i)+\sum\limits_{
     \begin{array}[b]{c}
        i=0 \\ 
        \theta\mid i
     \end{array}
}^{m-1}f_{ex}\left(\frac{i}{\theta}\right)\right](mod\ n).
\end{equation*}
\label{theorem}
\end{theorem}
\begin{proof}
Suppose the node sequence of the Flower code $\mathscr{F}_{\mathscr{C}}: (n, \theta, \alpha, \rho)$ is $\left\langle s_m\right\rangle _{m=1}^{\rho\theta}$, with an internal jump function 
$f_{in}(m): \left\{1,2,...,\rho\theta\right\}\backslash\{\theta, 2\theta,\ldots,\rho\theta\}\rightarrow\N\cup \{0\}$
and an external jump function 
$f_{ex}(m): \left\{1, 2,...,\rho\right\}\rightarrow\N\cup \{0\}$.
By the Definition \ref{eq:example_left_right1} of node sequence $\left\langle s_m\right\rangle _{m=1}^{\rho\theta}$ of Flower code $\mathscr{F}_{\mathscr{C}}: (n, \theta, \alpha, \rho)$, we have $s_m\left(\in\Omega_n\right)$ is the node on which packet is dropped and WLOG for $m=1$ one can take $s_1=1$. If $p$ is the index of term $s_p$ in $\left\langle s_m\right\rangle_{m=1}^{\rho\theta}$ then for $p \neq 1$ and $p\leq\rho\theta$, following cases are raised.
\begin{enumerate}
	\item If $\theta\mid p$ then $\theta^{th}$ packet is placed on $s_p$ for a cycle and first packet for the next cycle is placed on $s_{p+1}$. Hence there is external jump between packets dropped on node $s_p$ and node $s_{p+1}$. 
	Clearly $\frac{p}{\theta}$ gives the index number of jump completed at node $s_p$.
	\item If $\theta\nmid p$ then both packets dropped on nodes indexed $s_p$ and $s_{p+1}$, are from same cycle so there is internal jump between the both nodes. 
\end{enumerate}
Note that both cases can not fall on same node $s_m$ in node sequence $\left\langle s_m\right\rangle _{m=1}^{\rho\theta}$. Following case $1$ for $p=m-1$,  
\begin{equation}
s_{m} = s_{m-1} + f_{ex}\left(\frac{m-1}{\theta}\right)+1,
\label{recursion on external}
\end{equation}
where $\theta\mid(m-1)$ and $m\in\{1,2,\ldots\rho\theta\}$.
Again for case $2$, 
\begin{equation}
s_{m}=s_{m-1}+f_{in}\left(m-1\right)+1,
\label{recursion on internal}
\end{equation}
where $\theta\nmid(m-1)$ and $m\in\{1,2,\ldots\rho\theta\}$.

Hence one can have the following recursive equation by the above equations (\ref{recursion on external}) and (\ref{recursion on internal}) with boundary conditions $s_1=1$ and $s_i=0\ \forall i(\in\N)>\rho\theta$ on node sequence $\left\langle s_m\right\rangle_{m=1}^{\rho\theta}$.
\begin{equation}
 s_{m} = \left\{ \,
\begin{array}{ll}
1 & :\mbox{ if } m=1; \\
s_{t} + f_{ex}(t_{ex})+1 & :\mbox{ if } \theta\mid t, 1<m\leq \rho\theta; \\
s_{t}+f_{in}(t)+1 & :\mbox{ if } \theta\nmid t, 1<m\leq \rho\theta; \\
0 & :\mbox{ if } m>\rho\theta,
\end{array}
\right.\label{recursive equastion}
\end{equation}
 where $t_{ex} = \frac{m-1}{\theta}\mbox{ and }t = (m-1)$. Solving the recursive equation (\ref{recursive equastion}), one can get relation (\ref{Node sequence and jump}).
\end{proof}

\par Since each FR code can be represented by node sequence, dropping sequence and incidence matrix so each Flower code can also be represented by node sequence, dropping sequence and incidence matrix. The following lemmas establish the relation among those collectively. The lemmas are as follows.

\begin{lemma} Consider a dropping sequence $\left\langle d(m)\right\rangle_{m=1}^l$ for a Flower code $\mathscr{F}_{\mathscr{C}}: (n, \theta, \alpha, \rho)$. If $d(t)=1$ for any $t$ $(1\leq t\leq l)$ then binary node-packet distribution incidence matrix $M_{n \times \theta}$ = $[a_{ij}]_{n \times \theta}$ is given by $a_{ij}=1$, where
\begin{equation}
\begin{split}
& j = \left\{ \,
\begin{array}{ll}
\left[wt\left\langle d(m)\right\rangle_{m=1}^t\right](mod\ \theta) & :\mbox{ if } \theta\nmid\left[wt\left\langle d(m)\right\rangle_{m=1}^t\right]; \\
\theta & :\mbox{ if } \theta\mid\left[wt\left\langle d(m)\right\rangle_{m=1}^t\right].
\end{array}
\right. \\
& and \\
& i = \left\{ \,
\begin{array}{ll}
t\ (mod\ n) & :\mbox{ if } n\nmid t; \\
n & :\mbox{ if } n\mid t;
\end{array}
\right.
\end{split}
\end{equation}
\label{dropping sequence and matrix relation}
\end{lemma}

\begin{proof} Let $\left\langle d(m)\right\rangle_{m=1}^l$ be a dropping sequence for an FR code $\mathscr{C}: (n, \theta, \alpha, \rho)$ with $n$ nodes and $\theta$ packets. 
For $d(m)=1$ ($d(t)\in\left\langle d(m)\right\rangle_{m=1}^l$), index of the packet associated with the $d(t)$ is mapped to weight of subsequence $\left\langle d(m)\right\rangle_{m=1}^l$. Hence the lemma.

\end{proof}
\begin{remark} If Flower code $\mathscr{F}_{\mathscr{C}}: (n, \theta, \alpha, \rho)$ has a non-binary incidence matrix then one can calculate the incidence matrix by the Algorithm \ref{algo for incidence matrix} using dropping sequence.

\line(-1,0){20}\line(1,0){240}
\begin{algorithm}
Algorithm to compute node-packet distribution incidence matrix $M_{n \times \theta}$ for Flower code $\mathscr{F}_{\mathscr{C}}: (n, \theta, \alpha, \rho)$.
\label{algo for incidence matrix}
\end{algorithm}
\line(-1,0){20}\line(1,0){240}
\newline\textbf{REQUIRE} Dropping sequence $\left\langle d(m)\right\rangle_{m=1}^l$ of Flower code $\mathscr{F}_{\mathscr{C}}: (n, \theta, \alpha, \rho)$.
\newline\textbf{ENSURE} Node-packet distribution incidence matrix $M_{n \times \theta}$.
{\footnotesize\begin{enumerate}
\item Initially $\forall i,j$ set $a_{ij}=0$ and $t=1\ where\ a_{ij}$ is the element of $M_{n\times\theta},1\leq i\leq n,\ 1\leq j\leq\theta\ and\ 1\leq t\leq l.$ 
\item If $d(t)=0$ then jump to step $3$ and if $d(t)=1$ then calculate $i,j$ and set $a_{ij}=a_{ij}+1$ and go to step $3$, $where$
\begin{equation*}
\begin{split}
& i = \left\{ \,
\begin{array}{ll}
t\ (mod\ n) & :\mbox{ if } n\nmid t; \\
n & :\mbox{ if } n\mid t;
\end{array}
\right. and \\
& j = \left\{ \,
\begin{array}{ll}
\left[wt\left\langle d(m)\right\rangle_{m=1}^t\right](mod\ \theta) & :\mbox{ if } \theta\nmid\left[wt\left\langle d(m)\right\rangle_{m=1}^t\right]; \\
\theta & :\mbox{ if } \theta\mid\left[wt\left\langle d(m)\right\rangle_{m=1}^t\right].
\end{array}
\right.
\end{split}
\end{equation*}
\item If $t<l$ then set $t=t+1$ and go to step $2$ otherwise stop.
\end{enumerate}}
\line(-1,0){20}\line(1,0){240}

\par It is clear that Algorithm \ref{algo for incidence matrix} is the generalization of Lemma \ref{dropping sequence and matrix relation}.
\end{remark}


\begin{lemma} Consider a Flower code $\mathscr{F}_{\mathscr{C}}: (n, \theta, \alpha, \rho)$ with node sequence $\left\langle s_i\right\rangle_{i=1}^{\rho\theta}$.  Its dropping sequence is given by $\left\langle d(m)\right\rangle_{m=1}^l\ s.t.\ \forall\ i,1\leq i\leq\rho\theta$ 
\begin{equation*}
 d(m) = \left\{ \,
\begin{array}{ll}
1 & :\mbox{ if } m=\sum\limits_{j=1}^i\left(s_j-s_{j-1}\right)(mod\ n); \\
0 & :\mbox{ if } m\neq\sum\limits_{j=1}^i\left(s_j-s_{j-1}\right)(mod\ n),
\end{array}
\right.
\end{equation*}
where $s_o\ =\ 0$. 
\end{lemma}

\begin{proof} For a Flower code $\mathscr{F}_{\mathscr{C}}: (n, \theta, \alpha, \rho)$, consider a dropping sequence $\left\langle d(m)\right\rangle_{m=1}^l$. By Definition \ref{node sequence}, two packets with consecutive indexes are associated with $s_i$ and $s_{i+1}$ for some $i$. Using Definition \ref{dropping sequence}, one can find that $[s_{i+1}-s_i](mod\ n)$ number of zeros exist between two consecutive $1$'$s$. In particular, the $1$'$s$ are associative with the two packets. It proves the lemma.
\end{proof}

\begin{lemma} Consider a Flower code $\mathscr{F}_{\mathscr{C}}: (n, \theta, \alpha, \rho)$ with dropping sequence $\left\langle d(m)\right\rangle_{m=1}^l$. Its node sequence is given by $\left\langle s_i\right\rangle_{i=1}^{\rho\theta}\ s.t. \forall\ t\ (1\leq t\leq l)\ with\ d(t)=1$,
\begin{equation*}
 i = wt\left\langle d(m)\right\rangle_{m=1}^t 
\end{equation*}
 and 
\begin{equation*}
 s_i = \left\{ \,
\begin{array}{ll}
t\ (mod\ n) & :\mbox{ if } n\nmid t; \\
n & :\mbox{ if } n\mid t.
\end{array}
\right. \\
\end{equation*}
\end{lemma}

\begin{proof}
Using Definition \ref{node sequence} and Definition \ref{dropping sequence}, one can easily prove the lemma by observing that the weight of subsequence $\left\langle d(m)\right\rangle_{m=1}^t$ is associated with the packet for $d(t)=1$. 
\end{proof}

These sequence construction approaches suggest that one can construct an FR code from any arbitrary finite binary sequence by treating it as a characteristic sequence of dropping packets on the nodes. This can be done in different ways by choosing appropriate sequences of symbols. Hence, an arbitrary finite binary sequence $\left\langle \chi(m)\right\rangle_{m=1}^{\ell}\ (\chi(m)\in\Z_2)$ with length $\ell(\in\N)$ can be defined as a characteristic sequence for an FR code. In the sequence, value $\chi(m)=1$ represents to drop (value $\chi(m)=0$ not to drop) a packet on certain node.
Hence a more general Flower code $\mathscr{F}_{\mathscr{C}}: (n, \theta, \alpha , \rho)$ can be defined as follows.
\begin{definition} (Flower code): For $n$ nodes, $\theta$ packets and an arbitrary binary sequence $\left\langle \chi(m)\right\rangle_{m=1}^{\ell}\ (\chi(m)\in\Z_2)$ of length $\ell(\in\N)$ (treated as characteristic sequence indexed by $m$), Flower code $\mathscr{F}_{\mathscr{C}}: (n, \theta, \alpha , \rho)$ can be defined as a system in which packet indexed by $m(mod\ \theta)$ is dropped on node indexed by $m(mod\ n)$ iff $\chi(m)=1$, where $P_0$ and $U_0$ are mapped to packet $P_{\theta}$ and node $U_n$ respectively. Clearly $\alpha_i=\sum_{m=0}^{\left\lfloor \frac{\ell-i}{n}\right\rfloor}\chi(i+nm)$ and $\rho_j=\sum_{m=0}^{\left\lfloor \frac{\ell-j}{\theta}\right\rfloor}\chi(j+\theta m)$.
\label{Flower}
\end{definition}
\begin{example}
For given $n=4$, $\theta=5$ and a binary characteristic sequence $\left\langle \chi(m)\right\rangle_{m=1}^{16}$ = $\left\langle 1,1,0,1,0,0,1,1,1,1,1,1,\right.$ $\left. 0,1,0,1\right\rangle$, one can find Node packet distribution of Flower code $\mathscr{F}_{\mathscr{C}}:\ (4, 5, 3, 2)$ as shown in Table \ref{3}.

\begin{table}[ht]
\caption{Node-Packet Distribution for Flower code $\mathscr{F}_{\mathscr{C}}:\ (4, 5, 4, 3)$.}
\centering 
\begin{tabular}{|c||c|c|}
\hline
\textbf{Nodes $U_i$} &\textbf{Packets distribution} & \textbf{Repair degree $d_i$} \\[0.5ex] 
\hline\hline
$U_1$& $P_1,\ P_4$            & 2\\ \hline
$U_2$& $P_2,\ P_4\ P_5$       & 3\\ \hline 
$U_3$& $P_1,\ P_2,\ P_5$      & 3\\ \hline 
$U_4$& $P_1,\ P_2,\ P_3,\ P_4$& 4\\ \hline 
\end{tabular}
\label{3}
\end{table}
\end{example}

\begin{remark}
Flower codes $\mathscr{F}_{\mathscr{C}}: (n, \theta, \alpha , \rho)$ (Definition \ref{Flower}) with different packet replication factor $\rho_j$ $(j\in\Omega_{\theta})$ can also be constructed using any binary sequence $\left\langle \chi(m)\right\rangle_{m=1}^{\ell}$.
\end{remark}
\section{Conclusion}
This paper introduces a novel class of FR codes based on sequences. It will be an interesting future task to find a condition when Flower codes are optimal (in sense of the capacity). Our work opens a nice connection of FR codes with well known area of sequences. It would be an interesting future task to study FR codes by using well known class of sequences.  
\bibliographystyle{IEEEtran}
\bibliography{cloud}
\end{document}